\documentclass[graybox]{svmult}

\usepackage{type1cm}        
\usepackage{makeidx}         
\usepackage{graphicx}        
\usepackage{multicol}        
\usepackage[bottom]{footmisc}

\usepackage{newtxtext}       %
\usepackage{newtxmath}       
\usepackage{pgfarrows,tikz}
\usetikzlibrary{shapes,snakes,arrows,automata}

\newcommand\grp[2]{#1^{<#2>}}
\newcommand\subproj{\,\triangleleft\,}
\newcommand\simu{\leq}
\newcommand\bloc[1]{B_{\vec{#1}}}
\newcommand\Z{\mathbb{Z}}
\newcommand\N{\mathbb{N}}

\newcommand\restr[2]{{
  \left.\kern-\nulldelimiterspace 
  #1 
  \vphantom{\big|} 
  \right|_{#2} 
  }}

\begin{document}

\title*{The Mirage of Universality in Cellular Automata}
\author{Guillaume Theyssier}
\institute{Guillaume Theyssier \at CNRS, Universit\'e Aix-Marseille, \email{guillaume.theyssier@cnrs.fr}}

%
%
\maketitle

\newcommand\myabstract{This note is a survey of examples and results about cellular automata with the purpose of recalling that there is no 'universal' way of being computationally universal. In particular, we show how some cellular automata can embed efficient but bounded computation, while others can embed unbounded computations but not efficiently. We also study two variants of Boolean circuit embedding, transient versus repeatable simulations, and underline their differences. Finally we show how strong forms of universality can be hidden inside some seemingly simple cellular automata according to some classical dynamical parameters.}
\abstract*{\myabstract}

\abstract{\myabstract}

\section{Eric, the collector}

 The present note responds to an invitation to contribute to a book at the occasion of Eric Goles 70th birthday. Before diving into the scientific content, I should say a few words about Eric and the motivation behind this note.

Anyone knowing Eric certainly had the pleasure to listen to some of his colorful anecdotes (I certainly did). He owns a large collection, large enough to adapt to a wide variety of listeners and circumstances. The collection is in fact twice as large, because each anecdote, usually told to an international audience, is doubled with a more confidential Chilean version full of slang words. Eric's pleasure of telling stories is obvious, he has generously shared his collection, but nobody has listened to the same sequence of anecdotes and we all end up with a different global picture, much like the adventurous readers of the antinovel of Cort\'azar.

The collection of models and systems studied by Eric in its numerous scientific publications is equally striking. It abounds in small examples that are carefully analyzed and shown to capture important phenomena. It connects different points of view and different communities of researchers. It seems to never end up in the exact same theoretical framework and invites us to think about details that make a difference. In short, there is an anti-Bourbakist quality to it. 

At the heart of this scientific collection (at least from what I can tell from my collaboration with Eric), there is the question of the computational universality of small dynamical systems, and how it manifests itself in the complexity of various associated decision problems.
Computational universality of dynamical systems is a topic that might seem boring to the classical computer scientist (after all Turing showed the existence of a universal machine in the 1930s) and not serious for the dynamical systems community (this is not real maths\footnote{It should be noted however that a growing trend in symbolic dynamics has shown the importance of computability considerations. Some of these results were even published in real math journals... }). Part of the problem is that this kind of research is endangered by what I would call the \emph{mirage of universality}: the illusion that there must be universal consequences to the fact of being ``computationally universal'' independently of the precise definition used, and that such a statement, even given without technical details, gives information by itself. Pursuing this mirage, one is tempted to put forward vague theorem statements and hide the concrete mathematical result in the proofs (or sketch of). To make an analogy, no paper in computational complexity would use theorem statements like ``Problem X is hard'' and then, hidden in the proof details, unveil the definition of ``hard''. On the contrary, computational complexity theory has been extremely fruitful by putting forward a vast ``zoo'' of precisely defined complexity classes, often with a corresponding notion of reduction.

Of course, there is a lot to say and a lot has already been said about the mathematical formalization of computational universality in dynamical systems, but my intention here is clearly not to start a comprehensive survey on the topic \cite{DelvenneKB06,OllingerJAC}. Instead, I would like to invite the reader to a quick tour of examples and properties that break this mirage of universality. Most of them were encountered or established during my collaboration with Eric, and I hope this note can give a clue about the richness of Eric's scientific collection.\\

\textbf{Content of the note: } To simplify exposition, I chose to restrict to (classical) cellular automata and tackle three main topics in three separate sections. Each topic shows examples of ``computationally universal'' cellular automata that, in some sense, do not behave as expected, or pair of examples that behave differently with respect to some parameter:

\begin{itemize}
\item efficient versus unbounded computations: how some cellular automata are able to embed one type of computations but not the other;
\item transient versus repeatable circuit simulations: about the existence of (at least) two fundamentally different ways to simulate Boolean circuits in cellular automata, and their consequences;
\item hidden universality: how cellular automata might seem 'simple' according to some parameter despite being actually universal.
\end{itemize}
Before starting, some standard definition are given below to set up our framework.

\section{Standard definitions and notations}

For any finite set $Q$ (the alphabet) and positive integer $d$ (the dimension), we consider the space of configurations ${Q^{\Z^d}}$, \textit{i.e.} the set of maps giving a state from $Q$ to each position in the lattice $\Z^d$. The state of configuration ${c\in Q^{\Z^d}}$ at position ${z\in\Z^d}$ will be denoted either ${c(z)}$ or ${c_z}$.
\newcommand\zd{{\Z^d}}

\newcommand\ball[1]{\mathcal{B}({#1})}
For ${n\in\N}$, let $\ball{n}$ be the set of positions of ${\Z^d}$ of maximum norm at most $n$: 
\[\ball{n} = \{z\in\Z^d : \|z\|_\infty\leq n\}.\]
Then for any ${u\in Q^{\ball{n}}}$, we define the \emph{cylinder set} $[u]$ centered on cell 0 by:
\[[u] = \{c\in Q^{\Z^d} : \forall z\in\ball{n}, c_z=u_z\}.\]

These cylinder sets can be chosen as a base of open sets of the space $Q^\zd$ endowing it with a compact topology \cite{kurkabook}. Equivalently, the same topology can be defined by the Cantor distance: 
\[\delta(c,c') = 2^{-\min \{\|z\|_\infty : c_z\neq c'_z\}}.\]

A cellular automaton of dimension $d$ and state set $Q$ is a map $F$
acting continuously on configurations and translation invariant way.
Equivalently (Curtis-Lyndon-Heldund theorem \cite{hedlund}), it can be
defined locally by a neighborhood $V$ (a finite subset of $\zd$) and a
local transition map ${f:Q^V\rightarrow Q}$ as follows:
\[\forall z\in\zd,\quad F(c)_z = f\bigl(\restr{c}{z+V}\bigr)\]
where ${\restr{c}{z+V}}$ denotes the map ${z'\in V\mapsto c_{z+z'}}$.

The \emph{radius} of $F$ is the smallest integer $r$ such that ${V\subseteq\ball{r}}$ where $V$ is some neighborhood for which there is a local map $f_V:Q^V\rightarrow Q$ defining $F$ as above. $F$ induces an action on finite patterns as follows. For any ${n\in\N}$ and any ${u\in Q^{\ball{n+r}}}$, ${F(u)}$ is the finite pattern ${v\in Q^{\ball{n}}}$ obtained by application of $f$ on $u$ at each position from $\ball{n}$, \emph{i.e.\/} such that 
\[\forall c\in[u], F(c)\in[v].\]

We are now going to define a notion of universality for cellular automata. We choose this one for two reasons: first it is one of the strongest form of universality and will serve us as a benchmark in the following, and second, it is intrinsic to the model of cellular automata and make no reference to other models of computation (for more details, see \cite{OllingerUnivhistory,bulk1,Delorme2011b}).

This notion, called \emph{intrinsic universality}, is based on a notion of (intrinsic) simulation that is defined through two ingredients \cite{Delorme2011b,bulk1}.

The first ingredient is a notion of cell-wise simulation that works by restriction to a sub alphabet and then projection onto the target alphabet.
To be more precise let $F$ and $G$ be cellular automata of dimension $d$. 
We denote by ${F\subproj G}$ the fact that $F$
is obtained from $G$ by cell-wise restriction and projection,
formally: ${\exists \pi : Q\subseteq Q_G\rightarrow Q_F}$ surjective such
that for all ${c\in Q^{\zd}}$ 
\[\overline{\pi}\circ G(c) = F \circ \overline{\pi}(c)\]
where ${\overline{\pi} : Q^\zd\rightarrow Q_F^\zd}$ is the
cell-wise application of $\pi$. In the language of dynamical systems, ${(F,Q_F^{\zd})}$ is a factor of ${(G,Q^{\zd})}$ which is a sub-system of ${(G,Q_G^{\zd})}$.

Now we add the second ingredient, \emph{rescaling}, that allows to turn a cell-wise simulation into a simulation that works by blocks: blocks of cell of the first CA are simulated by blocks of cell of the second CA. Given a rectangular shape ${\vec{m}=(m_1,\ldots,m_d)}$ and some
alphabet $Q$ we define the bloc recoding map $\bloc{m}$ from $Q^\zd$ to ${\bigl(Q^{m_1m_2\cdots m_d}\bigr)^\zd}$ by:
\begin{align*}
  \bloc{m}(x)(z_1,\ldots,z_d) = \bigl(&x(m_1z_1,\ldots,m_dz_d),\ldots,\\
                              &x(m_1z_1+m_1-1,m_2z_2,\ldots,m_dz_d),\\
                              &x(m_1z_1,m_2z_2+1,m_3z_3,\ldots,m_dz_d),\ldots\\
                              &x(m_1z_1+m_1-1,\ldots,m_dz_d+m_d-1)\bigr)
\end{align*}
It is a bijection that recodes any configurations by blocks of shape $\vec{m}$.
Now if $t$ is a positive integer and $\vec{z}\in\zd$, we
define the rescaling of $F$ of parameters $\vec{m}$ and $t$
as the CA ${\grp{F}{\vec{m},t} = \bloc{m} \circ
F^t \circ \bloc{m}^{-1}}$.

We finally say that $G$ \emph{simulates} $F$, denoted ${F\simu G}$, if there
are parameters $\vec{m}$,$t$,$\vec{m'}$ and $t'$ 
such that
${\grp{F}{\vec{m},t}\subproj\grp{G}{\vec{m'},t'}}$. We also say that $G$ \emph{strongly simulates} $F$ if there are parameters $\vec{m}$ and $t$ such that 
${F\subproj\grp{G}{\vec{m},t}}$. 
Then, a CA $G$ is \emph{intrinsically universal} if for any CA $F$ we have
${F\simu G}$. It can be shown that an intrinsically universal CA can in fact \emph{strongly simulate} any CA \cite{bulk1}.

\newcommand\PTIME{\mathrm{P}}
\newcommand\PSPACE{\mathrm{PSPACE}}
\newcommand\NL{\mathrm{NLOGSPACE}}
\newcommand\LOG{\mathrm{LOGSPACE}}

Finally, we assume the reader is familiar with basic notions and results of computability and complexity theory. We will use the following standard classes of decision problems:
\begin{itemize}
\item $\PTIME$ is the set of problems which can be solve be a deterministic Turing machine in polynomial time;
\item $\NL$ is the set of problems which can be solve be a non-deterministic Turing machine in logarithmic space;
\item $\Sigma_1^0$ is the set of recursively enumerable problems (which contains the halting problem).
\end{itemize}

Without explicit mention and when speaking about $\PTIME$-completeness we consider LOGSPACE reductions. When speaking about $\Sigma_1^0$-completeness we usually consider many-one reductions.

\section{Efficient vs. unbounded computations}

It is well-known that, besides the reference model of Turing machines, there are other ones that fundamentally differ because they either only allow efficient but bounded computation (like Boolean circuits) or unbounded but slow computations (like Minsky machines) \cite{Minsky}. We would like to illustrate this aspect in the framework of cellular automata in a precise manner. To simplify, we restrict to dimension 1 in this section. We first define two classical problems associated to any CA which will serve as canonical indicators for both aspects mentioned above: efficiency and unboundedness of computations.

\newcommand\logred{\leq_{log}}
\newcommand\PRED[1]{\mathrm{PRED}_{#1}}
\newcommand\LIMIT[1]{\mathrm{LIMIT}_{#1}}
\newcommand\TIME[1]{\mathrm{TIME}_{#1}}
\newcommand\CYREACH[1]{\mathrm{UBPRED}_{#1}}

The first one is about short-term predictability within a bounded time range and provides a fine-grained complexity measurement within class $\PTIME$.

\begin{definition}
  Let $F$ be any CA of radius $r$ and alphabet $Q$.
  The prediction problem $\PRED{F}$ is defined as follows:
  \begin{itemize}
  \item input: ${t>0}$ and ${u\in Q^{\ball{rt}}}$
  \item output: ${F^t(u)\in Q}$.
  \end{itemize}
\end{definition}

The second one asks for a prediction about an unbounded future and provides a coarse-grained complexity measure allowed to cross the decidable barrier. It could be refined in many ways as in the definition of universality for dynamical symbolic systems from \cite{DelvenneKB06}. We prefer to keep it simple for the clarity of exposition. We say a configuration $c\in Q^\Z$ is eventually bi-periodic if it is eventually periodic to the left and eventually periodic to the right, said differently if it is of the form ${{}^\infty u_L\cdot u\cdot u_R^\infty}$ where $u_L$, $u$ and $u_R$ are finite words.

\begin{definition}
  Let $F$ be any CA of radius $r$ and dimension $1$.
  The reachability problem $\CYREACH{F}$ is defined as follows:
  \begin{itemize}
  \item input: an eventually bi-periodic configuration $c={}^\infty u_L\cdot u\cdot u_R^\infty$ and a state $q$.
  \item output: decide whether there is ${t\in\N}$ such that ${F^t(c)_0 = q}$.
  \end{itemize}
\end{definition}

One of the well-know results of computational universality in cellular automata is about elementary rule 110 given by the local rule 
${\delta : \{0,1\}^3\rightarrow \{0,1\}}$ with 
\[\delta(x,y,z) = (1-xyz)\cdot\max(y,z).\]

It is interesting to note that the first proof of computational universality of this cellular automaton due to M. Cook \cite{cook110} was enough to prove undecidability of $\CYREACH{\delta}$ but did not give information about problem $\PRED{\delta}$. It is only later, by a strong improvement in one step the the reduction, that $\PRED{\delta}$ was proven to be $\PTIME$-complete \cite{woodsneary06}. The purpose of this section is precisely to make clear that there is generally no implication in either direction between the $\PTIME$-hardness of $\PRED{}$ and the undecidability of $\CYREACH{}$. 

\begin{definition}
\label{def:freezing}
A CA $F$ is a \emph{freezing CA} if, for some (partial) order $\leq$ on states, the state of any cell can only decrease, \textit{i.e.\/}
  \[F(c)_z\leq c_z\]
  for any configuration $c$ and any cell $z$.
\end{definition}

The definition above was introduced in \cite{GolOlThey15} in studied more in depth in \cite{corr/fbccca}. Similar cellular automata corresponding to bounded changes or bounded communications were also considered in the literature with the point of view language recognizers \cite{vollmar81,KutribM10a,CartonGR18}. Under the hypothesis that $\NL\neq \PTIME$, the following results show examples of cellular automata that can embed arbitrary unbounded computation, but not in an efficient way.

\begin{theorem}[Section 4.3 of \cite{corr/fbccca}]
  For any freezing CA $F$ of dimension 1, the problem ${\PRED{F}}$ is in $\NL$. There exists a 1D freezing CA $F$ such that ${\CYREACH{F}}$ is ${\Sigma_1^0}$-complete.
\end{theorem}

We are now going to build an example with the opposite computation embedding properties: as hard as it can be in the short term (it can embed efficiently bounded computations), but decidable in the long term (it can not embed unbounded computations). It is inspired from \cite[Example 7]{corr/fbccca} and consists in a simulation of some $\PTIME$-complete cellular automaton inside finite zones, with some head controlling the simulation and forced to move back and forth inside the zone and shrink it by one cell at each bounce on a boundary. The simulation is such that one step of the simulated cellular automaton is done at each pass so that there is only a quadratic slowdown (see Figure~\ref{fig:zigzag}).

\newcommand\convzz[1]{\mathcal{Z}_{#1}}

\newcommand\stateZ[2]{\draw[fill=white] (#1,#2) rectangle +(1,1);}
\newcommand\stateO[2]{\draw[fill=white!50!blue] (#1,#2) rectangle +(1,1);}
\newcommand\stateZO[2]{\draw[fill=green] (#1,#2) rectangle +(1,1);}
\newcommand\stateOO[2]{\draw[fill=gray] (#1,#2) rectangle +(1,1);}
\newcommand\stateZZO[2]{\draw (#1,#2) rectangle +(1,1); \draw (#1,#2)+(.5,.5) node {\tiny $<$};}
\newcommand\stateOZO[2]{\draw (#1,#2) rectangle +(1,1); \draw (#1,#2)+(.5,.5) node {\tiny $>$};}
\newcommand\stateZOO[2]{\draw[fill=white!50!red] (#1,#2) rectangle +(1,1);}
\newcommand\stateOOO[2]{\draw[fill=white] (#1,#2) rectangle +(1,1);}

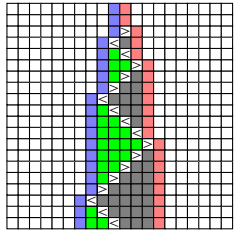
\begin{figure}
  \centering
        \begin{tikzpicture}[baseline=(current bounding box).center,scale=.15]
        \stateZ{0}{0}\stateZ{0}{1}\stateZ{0}{2}\stateZ{0}{3}\stateZ{0}{4}\stateZ{0}{5}\stateZ{0}{6}\stateZ{0}{7}\stateZ{0}{8}\stateZ{0}{9}\stateZ{0}{10}\stateZ{0}{11}\stateZ{0}{12}\stateZ{0}{13}\stateZ{0}{14}\stateZ{0}{15}\stateZ{0}{16}\stateZ{0}{17}\stateZ{0}{18}\stateZ{0}{19}
\stateZ{1}{0}\stateZ{1}{1}\stateZ{1}{2}\stateZ{1}{3}\stateZ{1}{4}\stateZ{1}{5}\stateZ{1}{6}\stateZ{1}{7}\stateZ{1}{8}\stateZ{1}{9}\stateZ{1}{10}\stateZ{1}{11}\stateZ{1}{12}\stateZ{1}{13}\stateZ{1}{14}\stateZ{1}{15}\stateZ{1}{16}\stateZ{1}{17}\stateZ{1}{18}\stateZ{1}{19}
\stateZ{2}{0}\stateZ{2}{1}\stateZ{2}{2}\stateZ{2}{3}\stateZ{2}{4}\stateZ{2}{5}\stateZ{2}{6}\stateZ{2}{7}\stateZ{2}{8}\stateZ{2}{9}\stateZ{2}{10}\stateZ{2}{11}\stateZ{2}{12}\stateZ{2}{13}\stateZ{2}{14}\stateZ{2}{15}\stateZ{2}{16}\stateZ{2}{17}\stateZ{2}{18}\stateZ{2}{19}
\stateZ{3}{0}\stateZ{3}{1}\stateZ{3}{2}\stateZ{3}{3}\stateZ{3}{4}\stateZ{3}{5}\stateZ{3}{6}\stateZ{3}{7}\stateZ{3}{8}\stateZ{3}{9}\stateZ{3}{10}\stateZ{3}{11}\stateZ{3}{12}\stateZ{3}{13}\stateZ{3}{14}\stateZ{3}{15}\stateZ{3}{16}\stateZ{3}{17}\stateZ{3}{18}\stateZ{3}{19}
\stateZ{4}{0}\stateZ{4}{1}\stateZ{4}{2}\stateZ{4}{3}\stateZ{4}{4}\stateZ{4}{5}\stateZ{4}{6}\stateZ{4}{7}\stateZ{4}{8}\stateZ{4}{9}\stateZ{4}{10}\stateZ{4}{11}\stateZ{4}{12}\stateZ{4}{13}\stateZ{4}{14}\stateZ{4}{15}\stateZ{4}{16}\stateZ{4}{17}\stateZ{4}{18}\stateZ{4}{19}
\stateZ{5}{0}\stateZ{5}{1}\stateZ{5}{2}\stateZ{5}{3}\stateZ{5}{4}\stateZ{5}{5}\stateZ{5}{6}\stateZ{5}{7}\stateZ{5}{8}\stateZ{5}{9}\stateZ{5}{10}\stateZ{5}{11}\stateZ{5}{12}\stateZ{5}{13}\stateZ{5}{14}\stateZ{5}{15}\stateZ{5}{16}\stateZ{5}{17}\stateZ{5}{18}\stateZ{5}{19}
\stateO{6}{0}\stateO{6}{1}\stateO{6}{2}\stateZ{6}{3}\stateZ{6}{4}\stateZ{6}{5}\stateZ{6}{6}\stateZ{6}{7}\stateZ{6}{8}\stateZ{6}{9}\stateZ{6}{10}\stateZ{6}{11}\stateZ{6}{12}\stateZ{6}{13}\stateZ{6}{14}\stateZ{6}{15}\stateZ{6}{16}\stateZ{6}{17}\stateZ{6}{18}\stateZ{6}{19}
\stateZO{7}{0}\stateZO{7}{1}\stateZZO{7}{2}\stateO{7}{3}\stateO{7}{4}\stateO{7}{5}\stateO{7}{6}\stateO{7}{7}\stateO{7}{8}\stateO{7}{9}\stateO{7}{10}\stateO{7}{11}\stateZ{7}{12}\stateZ{7}{13}\stateZ{7}{14}\stateZ{7}{15}\stateZ{7}{16}\stateZ{7}{17}\stateZ{7}{18}\stateZ{7}{19}
\stateZO{8}{0}\stateZZO{8}{1}\stateOO{8}{2}\stateOZO{8}{3}\stateZO{8}{4}\stateZO{8}{5}\stateZO{8}{6}\stateZO{8}{7}\stateZO{8}{8}\stateZO{8}{9}\stateZO{8}{10}\stateZZO{8}{11}\stateO{8}{12}\stateO{8}{13}\stateO{8}{14}\stateO{8}{15}\stateO{8}{16}\stateZ{8}{17}\stateZ{8}{18}\stateZ{8}{19}
\stateZZO{9}{0}\stateOO{9}{1}\stateOO{9}{2}\stateOO{9}{3}\stateOZO{9}{4}\stateZO{9}{5}\stateZO{9}{6}\stateZO{9}{7}\stateZO{9}{8}\stateZO{9}{9}\stateZZO{9}{10}\stateOO{9}{11}\stateOZO{9}{12}\stateZO{9}{13}\stateZO{9}{14}\stateZO{9}{15}\stateZZO{9}{16}\stateO{9}{17}\stateO{9}{18}\stateO{9}{19}
\stateOO{10}{0}\stateOO{10}{1}\stateOO{10}{2}\stateOO{10}{3}\stateOO{10}{4}\stateOZO{10}{5}\stateZO{10}{6}\stateZO{10}{7}\stateZO{10}{8}\stateZZO{10}{9}\stateOO{10}{10}\stateOO{10}{11}\stateOO{10}{12}\stateOZO{10}{13}\stateZO{10}{14}\stateZZO{10}{15}\stateOO{10}{16}\stateOZO{10}{17}\stateZOO{10}{18}\stateZOO{10}{19}
\stateOO{11}{0}\stateOO{11}{1}\stateOO{11}{2}\stateOO{11}{3}\stateOO{11}{4}\stateOO{11}{5}\stateOZO{11}{6}\stateZO{11}{7}\stateZZO{11}{8}\stateOO{11}{9}\stateOO{11}{10}\stateOO{11}{11}\stateOO{11}{12}\stateOO{11}{13}\stateOZO{11}{14}\stateZOO{11}{15}\stateZOO{11}{16}\stateZOO{11}{17}\stateOOO{11}{18}\stateOOO{11}{19}
\stateOO{12}{0}\stateOO{12}{1}\stateOO{12}{2}\stateOO{12}{3}\stateOO{12}{4}\stateOO{12}{5}\stateOO{12}{6}\stateOZO{12}{7}\stateZOO{12}{8}\stateZOO{12}{9}\stateZOO{12}{10}\stateZOO{12}{11}\stateZOO{12}{12}\stateZOO{12}{13}\stateZOO{12}{14}\stateOOO{12}{15}\stateOOO{12}{16}\stateOOO{12}{17}\stateOOO{12}{18}\stateOOO{12}{19}
\stateZOO{13}{0}\stateZOO{13}{1}\stateZOO{13}{2}\stateZOO{13}{3}\stateZOO{13}{4}\stateZOO{13}{5}\stateZOO{13}{6}\stateZOO{13}{7}\stateOOO{13}{8}\stateOOO{13}{9}\stateOOO{13}{10}\stateOOO{13}{11}\stateOOO{13}{12}\stateOOO{13}{13}\stateOOO{13}{14}\stateOOO{13}{15}\stateOOO{13}{16}\stateOOO{13}{17}\stateOOO{13}{18}\stateOOO{13}{19}
\stateOOO{14}{0}\stateOOO{14}{1}\stateOOO{14}{2}\stateOOO{14}{3}\stateOOO{14}{4}\stateOOO{14}{5}\stateOOO{14}{6}\stateOOO{14}{7}\stateOOO{14}{8}\stateOOO{14}{9}\stateOOO{14}{10}\stateOOO{14}{11}\stateOOO{14}{12}\stateOOO{14}{13}\stateOOO{14}{14}\stateOOO{14}{15}\stateOOO{14}{16}\stateOOO{14}{17}\stateOOO{14}{18}\stateOOO{14}{19}
\stateOOO{15}{0}\stateOOO{15}{1}\stateOOO{15}{2}\stateOOO{15}{3}\stateOOO{15}{4}\stateOOO{15}{5}\stateOOO{15}{6}\stateOOO{15}{7}\stateOOO{15}{8}\stateOOO{15}{9}\stateOOO{15}{10}\stateOOO{15}{11}\stateOOO{15}{12}\stateOOO{15}{13}\stateOOO{15}{14}\stateOOO{15}{15}\stateOOO{15}{16}\stateOOO{15}{17}\stateOOO{15}{18}\stateOOO{15}{19}
\stateOOO{16}{0}\stateOOO{16}{1}\stateOOO{16}{2}\stateOOO{16}{3}\stateOOO{16}{4}\stateOOO{16}{5}\stateOOO{16}{6}\stateOOO{16}{7}\stateOOO{16}{8}\stateOOO{16}{9}\stateOOO{16}{10}\stateOOO{16}{11}\stateOOO{16}{12}\stateOOO{16}{13}\stateOOO{16}{14}\stateOOO{16}{15}\stateOOO{16}{16}\stateOOO{16}{17}\stateOOO{16}{18}\stateOOO{16}{19}
\stateOOO{17}{0}\stateOOO{17}{1}\stateOOO{17}{2}\stateOOO{17}{3}\stateOOO{17}{4}\stateOOO{17}{5}\stateOOO{17}{6}\stateOOO{17}{7}\stateOOO{17}{8}\stateOOO{17}{9}\stateOOO{17}{10}\stateOOO{17}{11}\stateOOO{17}{12}\stateOOO{17}{13}\stateOOO{17}{14}\stateOOO{17}{15}\stateOOO{17}{16}\stateOOO{17}{17}\stateOOO{17}{18}\stateOOO{17}{19}
\stateOOO{18}{0}\stateOOO{18}{1}\stateOOO{18}{2}\stateOOO{18}{3}\stateOOO{18}{4}\stateOOO{18}{5}\stateOOO{18}{6}\stateOOO{18}{7}\stateOOO{18}{8}\stateOOO{18}{9}\stateOOO{18}{10}\stateOOO{18}{11}\stateOOO{18}{12}\stateOOO{18}{13}\stateOOO{18}{14}\stateOOO{18}{15}\stateOOO{18}{16}\stateOOO{18}{17}\stateOOO{18}{18}\stateOOO{18}{19}
\stateOOO{19}{0}\stateOOO{19}{1}\stateOOO{19}{2}\stateOOO{19}{3}\stateOOO{19}{4}\stateOOO{19}{5}\stateOOO{19}{6}\stateOOO{19}{7}\stateOOO{19}{8}\stateOOO{19}{9}\stateOOO{19}{10}\stateOOO{19}{11}\stateOOO{19}{12}\stateOOO{19}{13}\stateOOO{19}{14}\stateOOO{19}{15}\stateOOO{19}{16}\stateOOO{19}{17}\stateOOO{19}{18}\stateOOO{19}{19}
      \end{tikzpicture}

  \caption{The shrinking zone trick behind the construction of $\convzz{F}$ (time goes from bottom to top).}
  \label{fig:zigzag}
\end{figure}

Let $F$ be any 1D CA on alphabet $Q$ with radius $1$ and local map ${\delta : Q^3\rightarrow Q}$. We define ${\convzz{F}}$ on alphabet ${R = \{b,b_+,e\}\cup Q'}$ with ${Q'=Q\times Q\times\{\leftarrow,\rightarrow,l,r\}}$ and radius $1$ as follows:
\begin{itemize}
\item $e$, the \emph{error state}, is a spreading state: any cell with $e$ in its neighborhood turns into state $e$; a configuration $c$ is \emph{valid} if $e$ never appears in its orbit;
\item  $b$, the blank state, never changes except in presence of the error state; $b_+$ becomes $b$ except in presence of the error state; a maximal connect component of cells in state $Q'$ is a \emph{working zone};
\item in a working zone, patterns of the form ${(x,y,r)(x',y',l)}$, or ${(x',y',l)(x,y,r)}$, or ${(x,y,z)(x',y',z')}$ with ${\{z,z'\}\subseteq\{\leftarrow,\rightarrow\}}$, or ${(x,y,r)(x',y',z)}$ or ${(x,y,z)(x',y',l)}$ with ${z\in \{\leftarrow,\rightarrow\}}$, are forbidden and generate an $e$ state when detected; therefore in a valid configuration and in each working zone there is at most one occurrence of a state of the form ${(x,y,\{\leftarrow,\rightarrow\})}$ called the \emph{head};
\item a cell without forbidden pattern (from previous item) and without head in its neighborhood doesn't change its state;
\item the movements and actions of the heads are as follows:
  \begin{itemize}
  \item inside a working zone, the head in state $\leftarrow$ moves left, the head in state $\rightarrow$ moves right; the local map $\delta$ is only applied the head moves left to right; precisely we have the following transitions:
    \begin{align*}
      (x,y,l) ,\ (x',y',\leftarrow) ,\ (x'',y'',r) &\mapsto (x',y',r)\\
      (x,y,l) ,\ (x',y',l) ,\ (x'',y'',\leftarrow) &\mapsto (x',y',\leftarrow)\\
      (x,y,l) ,\ (x',y',\rightarrow) ,\ (x'',y'',r) &\mapsto (x',y',l)\\
      (x,y,\rightarrow) ,\ (x',y',r) ,\ (x'',y'',r) &\mapsto (\delta(y,x',x''),x',\rightarrow)\\
    \end{align*}
  \item when a boundary of the working zone is reached, the head bounces, changes of direction and the working zone get shrinked by one cell; precisely we have the following transitions:
    \begin{align*}
      b,\ (x,y,l),\ (x',y',\leftarrow) &\mapsto (x,y,\leftarrow)\\
      b,\ (x,y,\leftarrow),\ (x',y',r) &\mapsto (x,y,\rightarrow)\\
      b,\ (x,y,\rightarrow),\ (x',y',r) &\mapsto (y,x,l)\\
      b,\ (x,y,l),\ (x',y',\rightarrow) &\mapsto b_+\\
      (x,y,\rightarrow),\ (x',y',r),\ b &\mapsto (x',y',\rightarrow)\\
      (x,y,l),\ (x',y',\rightarrow),\ b &\mapsto (x',y',\leftarrow)\\
      (x,y,l),\ (x',y',\leftarrow),\ b &\mapsto (x',y',r)\\
      (x,y,\leftarrow),\ (x',y',r),\ b &\mapsto b_+\\
    \end{align*}
    \textit{(note the swap between $x$ and $y$ in the third transition above to initialize the sequential application of $\delta$)}
  \item finally the head disappears in a working zone of size $1$, precisely: 
    \[b',(x,y,z),b'' \mapsto (x,y,r)\]
    for any ${\{b',b''\}\subseteq \{b,b_+\}}$.
  \end{itemize}
\end{itemize}

\begin{theorem}
  For any $F$, the problem ${\CYREACH{\convzz{F}}}$ is decidable in polynomial time. If $F$ is chosen so that $\PRED{F}$ is $\PTIME$-complete, then ${\PRED{\convzz{F}}}$ is $\PTIME$-complete.
\end{theorem}
\begin{proof}
  For the first part of the Theorem, let us consider an eventually bi-periodic configuration ${c={}^\infty u_L\cdot u\cdot u_R^\infty}$. There are four cases:
  \begin{itemize}
  \item $c$ is not a valid configuration, which means that it contains a working zone with a forbidden pattern. Since the forbidden pattern are locally detectable, such a forbidden pattern must be detected inside the finite word ${u_Lu_Lu_Luu_ru_Ru_R}$ (considering the worst case where $u_L$ or $u_R$ is of size $1$). Therefore, in time $t$ which is linear in the sizes of $u_L$, $u$ and $u_R$ we have ${\convzz{F}^t(c)_0=e}$;
  \item $c$ is a valid configuration and position $0$ belong to a finite working zone in $c$. Since the left and right boundary of this zone must belong either to $u$, or $u_L$ or $u_R$, the zone is of linear size and it gets completely shrinked in quadratic time, meaning that the state of position $0$ will no longer change after a quadratic time;
  \item $c$ is a valid configuration and position $0$ belongs to an infinite zone in $c$. In this case, the position of the head must belong to either $u$, $u_L$ or $u_R$ and the same for the eventual (unique) boundary of the zone. Therefore, after a linear time in the worst case, cell $0$ will never change again (for instance, the head comes from the right, bounces to the left boundary, crosses once more position $0$, but never comes back again);
  \item $c$ is a valid configuration but position $0$ does not belong to some working zone, then for any ${t\geq 1}$ we have ${\convzz{F}^t(c)_0=b}$.
  \end{itemize}
  We deduce that that after a quadratic time the state of cell $0$ does not change any more, so it is sufficient to simulate $\convzz{F}$ on $c$ for this number of states to solve problem ${\CYREACH{\convzz{F}}}$.
  
  For the second part of the Theorem, see \cite[Lemma 1 and Proposition 5]{corr/fbccca}.
\end{proof}

\section{Transient vs. repeatable circuit simulation}
In this section we focus on dimension 2 and simulation of Boolean circuit and logical gates by cellular automata. Showing how a cellular automaton can embed Boolean circuits is one of the common methods used to claim its Turing universality (see for instance \cite{liferokadur,banks,OllingerJAC}).

We are now going to describe two modes of simulation of a set of
logical gates by a cellular automaton, which were formalized in \cite{GolesMPT18}. The basic simulation mechanism behind both simulation
modes uses square blocks concatenated in a grid-like fashion. Each
such square block represents a part of a concrete Boolean circuit
(either a node or wire). The definition doesn't require any specific
way of representing information inside the blocks, just that the family of
blocks use coherent representation of information so that the Boolean
logic works when assembling them. More concretely, they communicate
information with their four neighbors (north, east, south, west) in
such a way that each one implements a Boolean function with at most 2
inputs and at most 2 outputs.

In the sequel all considered blocks will compute one of the following
maps (we represent them using type
${\{0,1\}^4\rightarrow\{0,1\}^4}$ in order to make explicit the
position of inputs and outputs among the neighbors in the order north,
east, south, west):
\begin{align*}
  \text{AND}(x,\ast,\ast,y) &= \bigl(0,\min(x,y),0,0\bigr)\\
  \text{OR}(x,\ast,\ast,y) &= \bigl(0,\max(x,y),0,0\bigr)\\
  \text{CROSS}(x,\ast,\ast,y) &= \bigl(0,y,x,0\bigr)\\
  \text{NOP}(\ast,\ast,\ast,\ast) &= \bigl(0,0,0,0\bigr)\\
  \text{FORK}(\ast,\ast,\ast,x) &= \bigl(0,x,x,0\bigr)\\
  \text{WIRE}_{i,o}\bigl(c\in\{0,1\}^4\bigr) &= k\in\{0,\ldots,3\}
  \mapsto
  \begin{cases}
    c(i) &\text{ if }k=o\\
    0 &\text{ else}
  \end{cases}
\end{align*}
for any ${i\neq o\in\{0,\ldots,3\}}$. Note that any function $f$ above
is such that \[f(0,0,0,0)=(0,0,0,0).\] We denote by ${Img(f)}$ the set
of 4-uple that can be obtained as an image of $f$. The WIRE$_{i,o}$
functions are just all the possible ways to read a bit on one side and
transmit it to another side. Together with the NOP and FORK
function they represent the basic planar wiring toolkit denoted $W$
in the sequel. The NOP gate is special in that one considers it has 4 inputs and 4 outputs.

The two circuit simulation modes share the same block representation
of circuit and information, but they differ in their requirement about
the dynamical evolution of blocks. In the first mode, called \emph{transient mode}, the gates can
be used only once and nothing is granted concerning their evolution afterwards. The second mode, called \emph{repeatable mode}, asks for each gate to go back to some acceptable state each time they are used so that they can be used again. Both modes require the
simulation to work in constant time.

Let ${\mathcal{G}\subseteq\{\text{AND,OR,CROSS}\}}$ be a set of gates. Let $F$
be a CA with states set $Q$ and $N>0$ be an integer. Consider a
set ${V\subseteq Q^{N\times N}}$ of patterns, the valid blocks, each
of which as a type $f_u$ where ${f\in \mathcal{G}\cup W}$ and ${u\in Img(f)}$
and such that, for any $f_u$, there is some block of $V$ of type
$f_u$. If a block ${B\in V}$ has type ${f_{(a,b,c,d)}}$ for some
$f$, we say it has \emph{north value} $a$, \emph{east value} $b$,
\emph{south value} $c$ and \emph{west value} $d$. Finally let
${\Delta>0}$ be some constant. A configuration is \emph{valid} if it
is a concatenation of valid blocks where output sides of a block must face input sides of its corresponding neighbors. Given a block ${B\in V}$ of type
$f_u$ in a valid configuration, we say that it \emph{makes the correct transition} if it becomes a block of type $f_v$ after $\Delta$ steps where $v = f(n,e,s,w)$
is the output of $f$ on the input read from surrounding blocks,
precisely: the block at the north of $B$ has south value $n$, the block at
the east of $B$ has west value $e$, etc.  

\textbf{Transient simulation.} We say that \emph{$F$ simulates the set of
  gates $\mathcal{G}$ in transient mode} with delay $\Delta$ and valid blocks $V$ if for any
valid configuration $c$, the configuration ${F^\Delta(c)}$ is valid
and for any $f\in \mathcal{G}\cup W$, any block of type $f_{(0,0,0,0)}$ in $c$
makes the correct transition.

\textbf{Repeatable simulation.} The simulation is \emph{repeatable} if any block in any valid configuration makes the correct transition.\\

Before stating some theorems, let us define a decision problem associated to any 2D cellular automaton that will serve as a benchmark to separate the two kinds of circuit simulation above.

\newcommand\CYCL[2]{\mathrm{CYCLE}^{#1}_{#2}}
\begin{definition}
  Let $F$ be any 2D CA of radius $r$ and alphabet $Q$, and $\phi$ a non-decreasing function
  such that ${1\leq \phi(n)\leq 2^{O(n)}}$
  The prediction problem $\CYCL{\phi}{F}$ is defined as follows:
  \begin{itemize}
  \item input: a periodic configuration $c$ of period $n\times n$
  \item output: is the length of the temporal cycle reached from $c$
    strictly greater than $\phi(n)$?
  \end{itemize}
\end{definition}

The first mode of simulation (the repeatable mode) is the strongest one, and is actually equivalent to intrinsic universality even if we use only monotone gates. In this case, although the definition does not explicitly provide crossing gates, it is always possible to realize a dynamical crossing and build arbitrary reusable bloc elements leading to intrinsic universality.

\begin{theorem}[\cite{GolesMPT18}]
  A 2D CA $F$ is intrinsically universal if and only if it can simulate a ${\text{AND},\text{OR}}$ circuitry in a repeatable way. In this case $\PRED{F}$ is $P$-complete, $\CYREACH{F}$ is $\Sigma_1^0$-complete and $\CYCL{\phi}{F}$ is $\PSPACE$-complete for some $\phi$.
\end{theorem}

To illustrate the difference between the two modes we shall use the symmetric signed majority cellular automaton: it is essentially a majority rule where the state of each neighboring cell can be inverted or not before evaluating majority, this being done according to a local invariant sign vector and in a symmetric way: if cell $z$ inverts the value of its neighbor $z'$, then $z'$ will also invert the value of $z$.
We use the von Neumann neighborhood ${V= \{(0,0), (0,1), (1,0), (0,-1), (-1,0)\}}$.
The symmetric signed majority cellular automaton $F_1$ is defined over state set ${Q = \{-1,1\}^6}$. To simplify notation, we will see each state ${q\in Q}$ as a pair ${(I(q),S(q))}$ where ${I(q)\in\{0,1\}}$ represent the \emph{inner state} and ${S(q)\in \{-1,1\}^{V}}$ is a \emph{sign vector} associating a sign to each neighbor of the von Neumann neighborhood. For any configuration $c\in Q^{\Z^2}$, any cell $z$ and any cell ${z'\in z+V}$ we define the symmetric weight ${w_{zz'}\in\{-1,1\}}$ as ${w_{zz'}=(S(c_z)(z'-z))(S(c_{z'})(z-z'))}$. We note that ${w_{zz'}=w_{z'z}}$, hence the name symmetric. $F_1$ is then defined as follows.

\[F_1(c)_z = (\alpha, S(c_z))\]
where
\[\alpha =  \begin{cases}
    1 & \textrm{ if } \sum_{z' \in z+V} w_{zz'}I(c_{z'}) >0, \\ -1 & \textrm{otherwise.}
\end{cases}
\]

The following theorem shows that transient simulations are strictly weaker than repeatable simulations, $F_1$ being an example capable of the former, but not the latter.

\begin{theorem}[\cite{GolesMPT18}]
    If a 2D CA $F$ can simulate a ${\{\text{AND},\text{OR},\text{CROSS}\}}$
  circuitry in transient mode, then its associated problem $\PRED{F}$ is  $P$-complete. $F_1$ defined above can simulate a ${\{\text{AND},\text{OR},\text{CROSS}\}}$
  circuitry in transient mode. However, $F_1$ is not intrinsically universal and such that the problem $\CYCL{\phi}{F_1}$ is in $\PTIME$ if $\phi\equiv 1$ and trivial else.
\end{theorem}

\section{Hidden universality}

When proving that some cellular automaton is computationally universal, it can be acceptable to avoid a precise definition of universality if the construction makes it clear enough. However, a precise definition seems necessary when ones want to show that some cellular automaton is \textbf{not} universal. To avoid formalism, one could be tempted to use a shorter path: prove that, according to some well-chosen parameter, the considered cellular automaton is too simple to be universal. The intuition is that a universal cellular automaton should have roughly the highest complexity for the parameter. This approach can be made precise and yield some proof tools of non-universality in some contexts \cite{ccca}. The purpose of this section is to recall that things can get counter-intuitive and such a parameter must be chosen carefully.

\subsection{Hidden Minsky machines simulation}

The \emph{limit set} of a cellular automaton $F$  is the nonempty closed subset
\[\Omega_F=\bigcap_{t\in\N}F^t(X).\]
It represents the set of configurations that may appear arbitrarily far in the evolution and the restriction of $F$ to $\Omega_F$ is often considered as the asymptotic dynamics of $F$. The limit language is the set of finite patterns that occur in some configuration of $\Omega_F$. It is not difficult to see that the limit language is always co-recursively enumerable. However, there are known examples of non-recursive ones \cite{langlim2}. The attentive reader of \cite{Culik89} has probably spotted the affirmation that universal cellular automata have a non-recursive limit set. Depending on the definition of universality, this affirmation can be false. The following theorem shows that arbitrary Minsky machine simulations can be realize while maintaining a simple limit set (see \cite{limuniv} for the precise definition of simulation).

\begin{theorem}[\cite{limuniv}]\label{thm:hidd-minsky-mach} For any Minsky machine $M$ there exists a CA of dimension $1$ that simulates $M$ but whose limit language is regular.
\end{theorem}

\subsection{Hidden intrinsic universality}

Following Theorem~\ref{thm:hidd-minsky-mach}, one can go one step further and hide intrinsic universality behind a simple limit set (at the price of a complexity increase from regular to $\NL$). 
The main trick of the next theorem is inspired from \cite{kari94-2}: adding to a given cellular automaton $F$ on alphabet $Q$, a firing squad component (see Figure~\ref{fig:firingsquad}) that is able to fill-in the limit set restricted to the $Q$ component, and therefore make it simple independently of $F$.

  \newcommand\st[4]{{\draw[fill=#4] (#2,#3)
    -- ++(1,0) -- ++(0,1)--++(-1,0)--cycle;\draw (#2,#3)+(.5,.5) node {\tiny #1};}}
\newcommand\pointB[2]{}
\newcommand\pointF[2]{\st{\#}{#1}{#2}{gray}}
\newcommand\pointS[2]{\st{\#'}{#1}{#2}{gray!50!white}}
\newcommand\pointSp[2]{\st{$\gamma$}{#1}{#2}{yellow!50!white}}
\newcommand\pointXXXXX[2]{\st{$L_1$}{#1}{#2}{blue!50!white}}
\newcommand\pointXXXXXX[2]{\st{$l_1$}{#1}{#2}{blue!50!white}}
\newcommand\pointXXXXXXX[2]{\st{$R_1$}{#1}{#2}{red!50!white}}
\newcommand\pointXXXXXXXX[2]{\st{$r_1$}{#1}{#2}{red!50!white}}
\newcommand\pointXXXXXXXXX[2]{\st{$l_2$}{#1}{#2}{blue!20!white}}
\newcommand\pointXXXXXXXXXX[2]{\st{$l_2$}{#1}{#2}{blue!70!white}}
\newcommand\pointXXXXXXXXXXX[2]{\st{$r_2$}{#1}{#2}{red!20!white}}
\newcommand\pointXXXXXXXXXXXX[2]{\st{$r_2$}{#1}{#2}{red!70!white}}
\newcommand\pointXXXXXXXXXXXXX[2]{\st{X}{#1}{#2}{green!50!white}}
\newcommand\pointXXXXXXXXXXXXXX[2]{\st{Y}{#1}{#2}{green!50!white}}
\newcommand\pointXXXXXXXXXXXXXXX[2]{\st{Z}{#1}{#2}{green!50!white}}

    \begin{figure}
      \centering
      \begin{tikzpicture}[scale=.25]
      \draw[->] (-4,0) -- node[midway,sloped,above] {time} (-4,20);
      \pointB{0}{0}\pointB{0}{1}\pointB{0}{2}\pointB{0}{3}\pointB{0}{4}\pointB{0}{5}\pointB{0}{6}\pointXXXXX{0}{7}\pointB{0}{8}\pointXXXXXXXX{0}{9}\pointB{0}{10}\pointB{0}{11}\pointB{0}{12}\pointB{0}{13}\pointXXXXXXXXX{0}{14}\pointXXXXXXXXXX{0}{15}\pointB{0}{16}\pointXXXXXXX{0}{17}\pointXXXXXXXXXXX{0}{18}\pointXXXXXXXXXXXX{0}{19}\pointB{0}{20}\pointB{0}{21}\pointB{0}{22}\pointXXXXXX{0}{23}\pointB{0}{24}\pointXXXXXXX{0}{25}\pointXXXXXXXXXXX{0}{26}\pointXXXXXXXXXXXXXXX{0}{27}\pointB{0}{28}\pointS{0}{29}\pointSp{0}{30}
      \pointB{1}{0}\pointB{1}{1}\pointB{1}{2}\pointB{1}{3}\pointB{1}{4}\pointB{1}{5}\pointXXXXX{1}{6}\pointB{1}{7}\pointB{1}{8}\pointB{1}{9}\pointXXXXXXXX{1}{10}\pointB{1}{11}\pointXXXXXXXXX{1}{12}\pointXXXXXXXXXX{1}{13}\pointB{1}{14}\pointB{1}{15}\pointB{1}{16}\pointB{1}{17}\pointXXXXXXX{1}{18}\pointB{1}{19}\pointXXXXXXXXXXX{1}{20}\pointXXXXXXXXXXXX{1}{21}\pointXXXXXX{1}{22}\pointB{1}{23}\pointB{1}{24}\pointB{1}{25}\pointXXXXXXXXXXXXX{1}{26}\pointB{1}{27}\pointF{1}{28}\pointS{1}{29}\pointSp{1}{30}
      \pointB{2}{0}\pointB{2}{1}\pointB{2}{2}\pointB{2}{3}\pointB{2}{4}\pointXXXXX{2}{5}\pointB{2}{6}\pointB{2}{7}\pointB{2}{8}\pointB{2}{9}\pointXXXXXXXXX{2}{10}\pointXXXXXXXXXXXXXX{2}{11}\pointB{2}{12}\pointB{2}{13}\pointB{2}{14}\pointB{2}{15}\pointB{2}{16}\pointB{2}{17}\pointB{2}{18}\pointXXXXXXX{2}{19}\pointB{2}{20}\pointXXXXXX{2}{21}\pointXXXXXXXXXXX{2}{22}\pointXXXXXXXXXXXX{2}{23}\pointB{2}{24}\pointXXXXX{2}{25}\pointXXXXXXXXX{2}{26}\pointXXXXXXXXXXXXXX{2}{27}\pointB{2}{28}\pointS{2}{29}\pointSp{2}{30}
      \pointB{3}{0}\pointB{3}{1}\pointB{3}{2}\pointB{3}{3}\pointXXXXX{3}{4}\pointB{3}{5}\pointB{3}{6}\pointB{3}{7}\pointXXXXXXXXX{3}{8}\pointXXXXXXXXXX{3}{9}\pointB{3}{10}\pointB{3}{11}\pointXXXXXXXX{3}{12}\pointB{3}{13}\pointB{3}{14}\pointB{3}{15}\pointB{3}{16}\pointB{3}{17}\pointB{3}{18}\pointB{3}{19}\pointXXXXXXXXXXXXX{3}{20}\pointB{3}{21}\pointB{3}{22}\pointB{3}{23}\pointF{3}{24}\pointS{3}{25}\pointS{3}{26}\pointS{3}{27}\pointF{3}{28}\pointS{3}{29}\pointSp{3}{30}
      \pointB{4}{0}\pointB{4}{1}\pointB{4}{2}\pointXXXXX{4}{3}\pointB{4}{4}\pointB{4}{5}\pointXXXXXXXXX{4}{6}\pointXXXXXXXXXX{4}{7}\pointB{4}{8}\pointB{4}{9}\pointB{4}{10}\pointB{4}{11}\pointB{4}{12}\pointXXXXXXXX{4}{13}\pointB{4}{14}\pointB{4}{15}\pointB{4}{16}\pointB{4}{17}\pointB{4}{18}\pointXXXXX{4}{19}\pointB{4}{20}\pointXXXXXXXX{4}{21}\pointXXXXXXXXX{4}{22}\pointXXXXXXXXXX{4}{23}\pointB{4}{24}\pointXXXXXXX{4}{25}\pointXXXXXXXXXXX{4}{26}\pointXXXXXXXXXXXXXXX{4}{27}\pointB{4}{28}\pointS{4}{29}\pointSp{4}{30}
      \pointB{5}{0}\pointB{5}{1}\pointXXXXX{5}{2}\pointB{5}{3}\pointXXXXXXXXX{5}{4}\pointXXXXXXXXXX{5}{5}\pointB{5}{6}\pointB{5}{7}\pointB{5}{8}\pointB{5}{9}\pointB{5}{10}\pointB{5}{11}\pointB{5}{12}\pointB{5}{13}\pointXXXXXXXX{5}{14}\pointB{5}{15}\pointB{5}{16}\pointB{5}{17}\pointXXXXX{5}{18}\pointB{5}{19}\pointXXXXXXXXX{5}{20}\pointXXXXXXXXXX{5}{21}\pointXXXXXXXX{5}{22}\pointB{5}{23}\pointB{5}{24}\pointB{5}{25}\pointXXXXXXXXXXXXX{5}{26}\pointB{5}{27}\pointF{5}{28}\pointS{5}{29}\pointSp{5}{30}
      \pointB{6}{0}\pointXXXXX{6}{1}\pointXXXXXXXXX{6}{2}\pointXXXXXXXXXX{6}{3}\pointB{6}{4}\pointB{6}{5}\pointB{6}{6}\pointB{6}{7}\pointB{6}{8}\pointB{6}{9}\pointB{6}{10}\pointB{6}{11}\pointB{6}{12}\pointB{6}{13}\pointB{6}{14}\pointXXXXXXXX{6}{15}\pointB{6}{16}\pointXXXXX{6}{17}\pointXXXXXXXXX{6}{18}\pointXXXXXXXXXX{6}{19}\pointB{6}{20}\pointB{6}{21}\pointB{6}{22}\pointXXXXXXXX{6}{23}\pointB{6}{24}\pointXXXXX{6}{25}\pointXXXXXXXXX{6}{26}\pointXXXXXXXXXXXXXX{6}{27}\pointB{6}{28}\pointS{6}{29}\pointSp{6}{30}
      \pointF{7}{0}\pointS{7}{1}\pointS{7}{2}\pointS{7}{3}\pointS{7}{4}\pointS{7}{5}\pointS{7}{6}\pointS{7}{7}\pointS{7}{8}\pointS{7}{9}\pointS{7}{10}\pointS{7}{11}\pointS{7}{12}\pointS{7}{13}\pointS{7}{14}\pointS{7}{15}\pointF{7}{16}\pointS{7}{17}\pointS{7}{18}\pointS{7}{19}\pointS{7}{20}\pointS{7}{21}\pointS{7}{22}\pointS{7}{23}\pointF{7}{24}\pointS{7}{25}\pointS{7}{26}\pointS{7}{27}\pointF{7}{28}\pointS{7}{29}\pointSp{7}{30}
      \pointB{8}{0}\pointXXXXXXX{8}{1}\pointXXXXXXXXXXX{8}{2}\pointXXXXXXXXXXXX{8}{3}\pointB{8}{4}\pointB{8}{5}\pointB{8}{6}\pointB{8}{7}\pointB{8}{8}\pointB{8}{9}\pointB{8}{10}\pointB{8}{11}\pointB{8}{12}\pointB{8}{13}\pointB{8}{14}\pointXXXXXX{8}{15}\pointB{8}{16}\pointXXXXXXX{8}{17}\pointXXXXXXXXXXX{8}{18}\pointXXXXXXXXXXXX{8}{19}\pointB{8}{20}\pointB{8}{21}\pointB{8}{22}\pointXXXXXX{8}{23}\pointB{8}{24}\pointXXXXXXX{8}{25}\pointXXXXXXXXXXX{8}{26}\pointXXXXXXXXXXXXXXX{8}{27}\pointB{8}{28}\pointS{8}{29}\pointSp{8}{30}
      \pointB{9}{0}\pointB{9}{1}\pointXXXXXXX{9}{2}\pointB{9}{3}\pointXXXXXXXXXXX{9}{4}\pointXXXXXXXXXXXX{9}{5}\pointB{9}{6}\pointB{9}{7}\pointB{9}{8}\pointB{9}{9}\pointB{9}{10}\pointB{9}{11}\pointB{9}{12}\pointB{9}{13}\pointXXXXXX{9}{14}\pointB{9}{15}\pointB{9}{16}\pointB{9}{17}\pointXXXXXXX{9}{18}\pointB{9}{19}\pointXXXXXXXXXXX{9}{20}\pointXXXXXXXXXXXX{9}{21}\pointXXXXXX{9}{22}\pointB{9}{23}\pointB{9}{24}\pointB{9}{25}\pointXXXXXXXXXXXXX{9}{26}\pointB{9}{27}\pointF{9}{28}\pointS{9}{29}\pointSp{9}{30}
      \pointB{10}{0}\pointB{10}{1}\pointB{10}{2}\pointXXXXXXX{10}{3}\pointB{10}{4}\pointB{10}{5}\pointXXXXXXXXXXX{10}{6}\pointXXXXXXXXXXXX{10}{7}\pointB{10}{8}\pointB{10}{9}\pointB{10}{10}\pointB{10}{11}\pointB{10}{12}\pointXXXXXX{10}{13}\pointB{10}{14}\pointB{10}{15}\pointB{10}{16}\pointB{10}{17}\pointB{10}{18}\pointXXXXXXX{10}{19}\pointB{10}{20}\pointXXXXXX{10}{21}\pointXXXXXXXXXXX{10}{22}\pointXXXXXXXXXXXX{10}{23}\pointB{10}{24}\pointXXXXX{10}{25}\pointXXXXXXXXX{10}{26}\pointXXXXXXXXXXXXXX{10}{27}\pointB{10}{28}\pointS{10}{29}\pointSp{10}{30}
      \pointB{11}{0}\pointB{11}{1}\pointB{11}{2}\pointB{11}{3}\pointXXXXXXX{11}{4}\pointB{11}{5}\pointB{11}{6}\pointB{11}{7}\pointXXXXXXXXXXX{11}{8}\pointXXXXXXXXXXXX{11}{9}\pointB{11}{10}\pointB{11}{11}\pointXXXXXX{11}{12}\pointB{11}{13}\pointB{11}{14}\pointB{11}{15}\pointB{11}{16}\pointB{11}{17}\pointB{11}{18}\pointB{11}{19}\pointXXXXXXXXXXXXX{11}{20}\pointB{11}{21}\pointB{11}{22}\pointB{11}{23}\pointF{11}{24}\pointS{11}{25}\pointS{11}{26}\pointS{11}{27}\pointF{11}{28}\pointS{11}{29}\pointSp{11}{30}
      \pointB{12}{0}\pointB{12}{1}\pointB{12}{2}\pointB{12}{3}\pointB{12}{4}\pointXXXXXXX{12}{5}\pointB{12}{6}\pointB{12}{7}\pointB{12}{8}\pointB{12}{9}\pointXXXXXXXXXXX{12}{10}\pointXXXXXXXXXXXXXXX{12}{11}\pointB{12}{12}\pointB{12}{13}\pointB{12}{14}\pointB{12}{15}\pointB{12}{16}\pointB{12}{17}\pointB{12}{18}\pointXXXXX{12}{19}\pointB{12}{20}\pointXXXXXXXX{12}{21}\pointXXXXXXXXX{12}{22}\pointXXXXXXXXXX{12}{23}\pointB{12}{24}\pointXXXXXXX{12}{25}\pointXXXXXXXXXXX{12}{26}\pointXXXXXXXXXXXXXXX{12}{27}\pointB{12}{28}\pointS{12}{29}\pointSp{12}{30}
      \pointB{13}{0}\pointB{13}{1}\pointB{13}{2}\pointB{13}{3}\pointB{13}{4}\pointB{13}{5}\pointXXXXXXX{13}{6}\pointB{13}{7}\pointB{13}{8}\pointB{13}{9}\pointXXXXXX{13}{10}\pointB{13}{11}\pointXXXXXXXXXXX{13}{12}\pointXXXXXXXXXXXX{13}{13}\pointB{13}{14}\pointB{13}{15}\pointB{13}{16}\pointB{13}{17}\pointXXXXX{13}{18}\pointB{13}{19}\pointXXXXXXXXX{13}{20}\pointXXXXXXXXXX{13}{21}\pointXXXXXXXX{13}{22}\pointB{13}{23}\pointB{13}{24}\pointB{13}{25}\pointXXXXXXXXXXXXX{13}{26}\pointB{13}{27}\pointF{13}{28}\pointS{13}{29}\pointSp{13}{30}
      \pointB{14}{0}\pointB{14}{1}\pointB{14}{2}\pointB{14}{3}\pointB{14}{4}\pointB{14}{5}\pointB{14}{6}\pointXXXXXXX{14}{7}\pointB{14}{8}\pointXXXXXX{14}{9}\pointB{14}{10}\pointB{14}{11}\pointB{14}{12}\pointB{14}{13}\pointXXXXXXXXXXX{14}{14}\pointXXXXXXXXXXXX{14}{15}\pointB{14}{16}\pointXXXXX{14}{17}\pointXXXXXXXXX{14}{18}\pointXXXXXXXXXX{14}{19}\pointB{14}{20}\pointB{14}{21}\pointB{14}{22}\pointXXXXXXXX{14}{23}\pointB{14}{24}\pointXXXXX{14}{25}\pointXXXXXXXXX{14}{26}\pointXXXXXXXXXXXXXX{14}{27}\pointB{14}{28}\pointS{14}{29}\pointSp{14}{30}
      \pointB{15}{0}\pointB{15}{1}\pointB{15}{2}\pointB{15}{3}\pointB{15}{4}\pointB{15}{5}\pointB{15}{6}\pointB{15}{7}\pointXXXXXXXXXXXXX{15}{8}\pointB{15}{9}\pointB{15}{10}\pointB{15}{11}\pointB{15}{12}\pointB{15}{13}\pointB{15}{14}\pointB{15}{15}\pointF{15}{16}\pointS{15}{17}\pointS{15}{18}\pointS{15}{19}\pointS{15}{20}\pointS{15}{21}\pointS{15}{22}\pointS{15}{23}\pointF{15}{24}\pointS{15}{25}\pointS{15}{26}\pointS{15}{27}\pointF{15}{28}\pointS{15}{29}\pointSp{15}{30}
      \pointB{16}{0}\pointB{16}{1}\pointB{16}{2}\pointB{16}{3}\pointB{16}{4}\pointB{16}{5}\pointB{16}{6}\pointXXXXX{16}{7}\pointB{16}{8}\pointXXXXXXXX{16}{9}\pointB{16}{10}\pointB{16}{11}\pointB{16}{12}\pointB{16}{13}\pointXXXXXXXXX{16}{14}\pointXXXXXXXXXX{16}{15}\pointB{16}{16}\pointXXXXXXX{16}{17}\pointXXXXXXXXXXX{16}{18}\pointXXXXXXXXXXXX{16}{19}\pointB{16}{20}\pointB{16}{21}\pointB{16}{22}\pointXXXXXX{16}{23}\pointB{16}{24}\pointXXXXXXX{16}{25}\pointXXXXXXXXXXX{16}{26}\pointXXXXXXXXXXXXXXX{16}{27}\pointB{16}{28}\pointS{16}{29}\pointSp{16}{30}
      \pointB{17}{0}\pointB{17}{1}\pointB{17}{2}\pointB{17}{3}\pointB{17}{4}\pointB{17}{5}\pointXXXXX{17}{6}\pointB{17}{7}\pointB{17}{8}\pointB{17}{9}\pointXXXXXXXX{17}{10}\pointB{17}{11}\pointXXXXXXXXX{17}{12}\pointXXXXXXXXXX{17}{13}\pointB{17}{14}\pointB{17}{15}\pointB{17}{16}\pointB{17}{17}\pointXXXXXXX{17}{18}\pointB{17}{19}\pointXXXXXXXXXXX{17}{20}\pointXXXXXXXXXXXX{17}{21}\pointXXXXXX{17}{22}\pointB{17}{23}\pointB{17}{24}\pointB{17}{25}\pointXXXXXXXXXXXXX{17}{26}\pointB{17}{27}\pointF{17}{28}\pointS{17}{29}\pointSp{17}{30}
      \pointB{18}{0}\pointB{18}{1}\pointB{18}{2}\pointB{18}{3}\pointB{18}{4}\pointXXXXX{18}{5}\pointB{18}{6}\pointB{18}{7}\pointB{18}{8}\pointB{18}{9}\pointXXXXXXXXX{18}{10}\pointXXXXXXXXXXXXXX{18}{11}\pointB{18}{12}\pointB{18}{13}\pointB{18}{14}\pointB{18}{15}\pointB{18}{16}\pointB{18}{17}\pointB{18}{18}\pointXXXXXXX{18}{19}\pointB{18}{20}\pointXXXXXX{18}{21}\pointXXXXXXXXXXX{18}{22}\pointXXXXXXXXXXXX{18}{23}\pointB{18}{24}\pointXXXXX{18}{25}\pointXXXXXXXXX{18}{26}\pointXXXXXXXXXXXXXX{18}{27}\pointB{18}{28}\pointS{18}{29}\pointSp{18}{30}
      \pointB{19}{0}\pointB{19}{1}\pointB{19}{2}\pointB{19}{3}\pointXXXXX{19}{4}\pointB{19}{5}\pointB{19}{6}\pointB{19}{7}\pointXXXXXXXXX{19}{8}\pointXXXXXXXXXX{19}{9}\pointB{19}{10}\pointB{19}{11}\pointXXXXXXXX{19}{12}\pointB{19}{13}\pointB{19}{14}\pointB{19}{15}\pointB{19}{16}\pointB{19}{17}\pointB{19}{18}\pointB{19}{19}\pointXXXXXXXXXXXXX{19}{20}\pointB{19}{21}\pointB{19}{22}\pointB{19}{23}\pointF{19}{24}\pointS{19}{25}\pointS{19}{26}\pointS{19}{27}\pointF{19}{28}\pointS{19}{29}\pointSp{19}{30}
      \pointB{20}{0}\pointB{20}{1}\pointB{20}{2}\pointXXXXX{20}{3}\pointB{20}{4}\pointB{20}{5}\pointXXXXXXXXX{20}{6}\pointXXXXXXXXXX{20}{7}\pointB{20}{8}\pointB{20}{9}\pointB{20}{10}\pointB{20}{11}\pointB{20}{12}\pointXXXXXXXX{20}{13}\pointB{20}{14}\pointB{20}{15}\pointB{20}{16}\pointB{20}{17}\pointB{20}{18}\pointXXXXX{20}{19}\pointB{20}{20}\pointXXXXXXXX{20}{21}\pointXXXXXXXXX{20}{22}\pointXXXXXXXXXX{20}{23}\pointB{20}{24}\pointXXXXXXX{20}{25}\pointXXXXXXXXXXX{20}{26}\pointXXXXXXXXXXXXXXX{20}{27}\pointB{20}{28}\pointS{20}{29}\pointSp{20}{30}
      \pointB{21}{0}\pointB{21}{1}\pointXXXXX{21}{2}\pointB{21}{3}\pointXXXXXXXXX{21}{4}\pointXXXXXXXXXX{21}{5}\pointB{21}{6}\pointB{21}{7}\pointB{21}{8}\pointB{21}{9}\pointB{21}{10}\pointB{21}{11}\pointB{21}{12}\pointB{21}{13}\pointXXXXXXXX{21}{14}\pointB{21}{15}\pointB{21}{16}\pointB{21}{17}\pointXXXXX{21}{18}\pointB{21}{19}\pointXXXXXXXXX{21}{20}\pointXXXXXXXXXX{21}{21}\pointXXXXXXXX{21}{22}\pointB{21}{23}\pointB{21}{24}\pointB{21}{25}\pointXXXXXXXXXXXXX{21}{26}\pointB{21}{27}\pointF{21}{28}\pointS{21}{29}\pointSp{21}{30}
      \pointB{22}{0}\pointXXXXX{22}{1}\pointXXXXXXXXX{22}{2}\pointXXXXXXXXXX{22}{3}\pointB{22}{4}\pointB{22}{5}\pointB{22}{6}\pointB{22}{7}\pointB{22}{8}\pointB{22}{9}\pointB{22}{10}\pointB{22}{11}\pointB{22}{12}\pointB{22}{13}\pointB{22}{14}\pointXXXXXXXX{22}{15}\pointB{22}{16}\pointXXXXX{22}{17}\pointXXXXXXXXX{22}{18}\pointXXXXXXXXXX{22}{19}\pointB{22}{20}\pointB{22}{21}\pointB{22}{22}\pointXXXXXXXX{22}{23}\pointB{22}{24}\pointXXXXX{22}{25}\pointXXXXXXXXX{22}{26}\pointXXXXXXXXXXXXXX{22}{27}\pointB{22}{28}\pointS{22}{29}\pointSp{22}{30}
      \pointF{23}{0}\pointS{23}{1}\pointS{23}{2}\pointS{23}{3}\pointS{23}{4}\pointS{23}{5}\pointS{23}{6}\pointS{23}{7}\pointS{23}{8}\pointS{23}{9}\pointS{23}{10}\pointS{23}{11}\pointS{23}{12}\pointS{23}{13}\pointS{23}{14}\pointS{23}{15}\pointF{23}{16}\pointS{23}{17}\pointS{23}{18}\pointS{23}{19}\pointS{23}{20}\pointS{23}{21}\pointS{23}{22}\pointS{23}{23}\pointF{23}{24}\pointS{23}{25}\pointS{23}{26}\pointS{23}{27}\pointF{23}{28}\pointS{23}{29}\pointSp{23}{30}
      \pointB{24}{0}\pointXXXXXXX{24}{1}\pointXXXXXXXXXXX{24}{2}\pointXXXXXXXXXXXX{24}{3}\pointB{24}{4}\pointB{24}{5}\pointB{24}{6}\pointB{24}{7}\pointB{24}{8}\pointB{24}{9}\pointB{24}{10}\pointB{24}{11}\pointB{24}{12}\pointB{24}{13}\pointB{24}{14}\pointXXXXXX{24}{15}\pointB{24}{16}\pointXXXXXXX{24}{17}\pointXXXXXXXXXXX{24}{18}\pointXXXXXXXXXXXX{24}{19}\pointB{24}{20}\pointB{24}{21}\pointB{24}{22}\pointXXXXXX{24}{23}\pointB{24}{24}\pointXXXXXXX{24}{25}\pointXXXXXXXXXXX{24}{26}\pointXXXXXXXXXXXXXXX{24}{27}\pointB{24}{28}\pointS{24}{29}\pointSp{24}{30}
      \pointB{25}{0}\pointB{25}{1}\pointXXXXXXX{25}{2}\pointB{25}{3}\pointXXXXXXXXXXX{25}{4}\pointXXXXXXXXXXXX{25}{5}\pointB{25}{6}\pointB{25}{7}\pointB{25}{8}\pointB{25}{9}\pointB{25}{10}\pointB{25}{11}\pointB{25}{12}\pointB{25}{13}\pointXXXXXX{25}{14}\pointB{25}{15}\pointB{25}{16}\pointB{25}{17}\pointXXXXXXX{25}{18}\pointB{25}{19}\pointXXXXXXXXXXX{25}{20}\pointXXXXXXXXXXXX{25}{21}\pointXXXXXX{25}{22}\pointB{25}{23}\pointB{25}{24}\pointB{25}{25}\pointXXXXXXXXXXXXX{25}{26}\pointB{25}{27}\pointF{25}{28}\pointS{25}{29}\pointSp{25}{30}
      \pointB{26}{0}\pointB{26}{1}\pointB{26}{2}\pointXXXXXXX{26}{3}\pointB{26}{4}\pointB{26}{5}\pointXXXXXXXXXXX{26}{6}\pointXXXXXXXXXXXX{26}{7}\pointB{26}{8}\pointB{26}{9}\pointB{26}{10}\pointB{26}{11}\pointB{26}{12}\pointXXXXXX{26}{13}\pointB{26}{14}\pointB{26}{15}\pointB{26}{16}\pointB{26}{17}\pointB{26}{18}\pointXXXXXXX{26}{19}\pointB{26}{20}\pointXXXXXX{26}{21}\pointXXXXXXXXXXX{26}{22}\pointXXXXXXXXXXXX{26}{23}\pointB{26}{24}\pointXXXXX{26}{25}\pointXXXXXXXXX{26}{26}\pointXXXXXXXXXXXXXX{26}{27}\pointB{26}{28}\pointS{26}{29}\pointSp{26}{30}
      \pointB{27}{0}\pointB{27}{1}\pointB{27}{2}\pointB{27}{3}\pointXXXXXXX{27}{4}\pointB{27}{5}\pointB{27}{6}\pointB{27}{7}\pointXXXXXXXXXXX{27}{8}\pointXXXXXXXXXXXX{27}{9}\pointB{27}{10}\pointB{27}{11}\pointXXXXXX{27}{12}\pointB{27}{13}\pointB{27}{14}\pointB{27}{15}\pointB{27}{16}\pointB{27}{17}\pointB{27}{18}\pointB{27}{19}\pointXXXXXXXXXXXXX{27}{20}\pointB{27}{21}\pointB{27}{22}\pointB{27}{23}\pointF{27}{24}\pointS{27}{25}\pointS{27}{26}\pointS{27}{27}\pointF{27}{28}\pointS{27}{29}\pointSp{27}{30}
      \pointB{28}{0}\pointB{28}{1}\pointB{28}{2}\pointB{28}{3}\pointB{28}{4}\pointXXXXXXX{28}{5}\pointB{28}{6}\pointB{28}{7}\pointB{28}{8}\pointB{28}{9}\pointXXXXXXXXXXX{28}{10}\pointXXXXXXXXXXXXXXX{28}{11}\pointB{28}{12}\pointB{28}{13}\pointB{28}{14}\pointB{28}{15}\pointB{28}{16}\pointB{28}{17}\pointB{28}{18}\pointXXXXX{28}{19}\pointB{28}{20}\pointXXXXXXXX{28}{21}\pointXXXXXXXXX{28}{22}\pointXXXXXXXXXX{28}{23}\pointB{28}{24}\pointXXXXXXX{28}{25}\pointXXXXXXXXXXX{28}{26}\pointXXXXXXXXXXXXXXX{28}{27}\pointB{28}{28}\pointS{28}{29}\pointSp{28}{30}
      \pointB{29}{0}\pointB{29}{1}\pointB{29}{2}\pointB{29}{3}\pointB{29}{4}\pointB{29}{5}\pointXXXXXXX{29}{6}\pointB{29}{7}\pointB{29}{8}\pointB{29}{9}\pointXXXXXX{29}{10}\pointB{29}{11}\pointXXXXXXXXXXX{29}{12}\pointXXXXXXXXXXXX{29}{13}\pointB{29}{14}\pointB{29}{15}\pointB{29}{16}\pointB{29}{17}\pointXXXXX{29}{18}\pointB{29}{19}\pointXXXXXXXXX{29}{20}\pointXXXXXXXXXX{29}{21}\pointXXXXXXXX{29}{22}\pointB{29}{23}\pointB{29}{24}\pointB{29}{25}\pointXXXXXXXXXXXXX{29}{26}\pointB{29}{27}\pointF{29}{28}\pointS{29}{29}\pointSp{29}{30}
      \pointB{30}{0}\pointB{30}{1}\pointB{30}{2}\pointB{30}{3}\pointB{30}{4}\pointB{30}{5}\pointB{30}{6}\pointXXXXXXX{30}{7}\pointB{30}{8}\pointXXXXXX{30}{9}\pointB{30}{10}\pointB{30}{11}\pointB{30}{12}\pointB{30}{13}\pointXXXXXXXXXXX{30}{14}\pointXXXXXXXXXXXX{30}{15}\pointB{30}{16}\pointXXXXX{30}{17}\pointXXXXXXXXX{30}{18}\pointXXXXXXXXXX{30}{19}\pointB{30}{20}\pointB{30}{21}\pointB{30}{22}\pointXXXXXXXX{30}{23}\pointB{30}{24}\pointXXXXX{30}{25}\pointXXXXXXXXX{30}{26}\pointXXXXXXXXXXXXXX{30}{27}\pointB{30}{28}\pointS{30}{29}\pointSp{30}{30}
      \pointB{31}{0}\pointB{31}{1}\pointB{31}{2}\pointB{31}{3}\pointB{31}{4}\pointB{31}{5}\pointB{31}{6}\pointB{31}{7}\pointXXXXXXXXXXXXX{31}{8}\pointB{31}{9}\pointB{31}{10}\pointB{31}{11}\pointB{31}{12}\pointB{31}{13}\pointB{31}{14}\pointB{31}{15}\pointF{31}{16}\pointS{31}{17}\pointS{31}{18}\pointS{31}{19}\pointS{31}{20}\pointS{31}{21}\pointS{31}{22}\pointS{31}{23}\pointF{31}{24}\pointS{31}{25}\pointS{31}{26}\pointS{31}{27}\pointF{31}{28}\pointS{31}{29}\pointSp{31}{30}
    \end{tikzpicture}
    \caption{J. Kari's firing squad trick: a synchronous apparition of $\gamma$ can be triggered arbitrarily far in time, thus allowing to complete the limit set on some component of states to the full-shift.}
    \label{fig:firingsquad}
  \end{figure}

\begin{theorem}[\cite{GuillonMT10}]
  There exists an intrinsically universal CA whose limit language is $\NL$.
\end{theorem}

Given ${n\in\N}$, the \emph{column factor} of width ${n}$ of $F$, ${\Sigma_n(F)}$, is the set of columns that can appear in space-time diagrams of $F$: 
\[\Sigma_n(F) = \bigl\{(u_t)_{t\in\N} : u_t\in Q^n, u_t= F^t(c)_{[1,n]}, c\in Q^{\Z^d}\bigr\}.\]
To ${\Sigma_n(F)}$ we associate its language of finite patterns ${L(\Sigma_n(F))}$ defined as the set of words ${u_t\cdots u_{t+k}}$ for some ${(u_t)_{t\in\N}\in\Sigma_n(F)}$ and ${t,k\in\N}$. 

The approach of \cite{DelvenneKB06} to define universality for general dynamical systems translates into the following in our settings. To $F$ we associate the model checking problem:
\begin{itemize}
\item \textbf{input:} $n$ and a regular language $L_n$ over alphabet $Q^n$,
\item \textbf{question:} decide whether $L_n$ intersects $L(\Sigma_n(F))$.
\end{itemize}

$F$ is \emph{BDK-universal} if its associated model checking problem is r.e.-complete. Like for limit sets, column factors can be filled up and thus simplified by increasing the alphabet starting from an arbitrarily complex cellular automaton.

\begin{theorem}[\cite{GuillonMT10}]
  There exists an intrinsically universal CA $F$ such that ${L(\Sigma_n(F))}$ are regular languages computable from $n$. In particular, such $F$ is not BDK-universal.
\end{theorem}

\bibliography{refs}
\bibliographystyle{plain}

\end{document}